\documentclass[letterpaper,UKenglish]{lipics-v2016}
 
\usepackage{microtype}

\title{On rescheduling due to machine disruption while to minimize the total weighted completion time%
\footnote{This work was partially supported by NSERC Canada, Brain Canada, NSF China and CSC China.}}
\titlerunning{Rescheduling due to machine disruption} 

\author[1,2]{Wenchang Luo}
\author[2,3]{Taibo Luo}
\author[2]{Randy Goebel}
\author[2]{Guohui Lin}
\affil[1]{Faculty of Science, Ningbo University.
  Ningbo, Zhejiang 315211, China.}
\affil[2]{Department of Computing Science, University of Alberta.
  Edmonton, Alberta T6G 2E8, Canada.
  \texttt{\{wenchang,taibo,rgoebel,guohui\}@ualberta.ca}}
\affil[3]{Business School, Sichuan University.
  Chengdu, Sichuan 610065, China.}
\authorrunning{Luo {\it et al.} version/\today} 

\Copyright{Wenchang Luo, Taibo Luo, Randy Goebel and Guohui Lin}

\subjclass{Dummy classification -- please refer to \url{http://www.acm.org/about/class/ccs98-html}}

\begin{document}

\maketitle

\begin{abstract}
We investigate a single machine rescheduling problem that arises from an unexpected machine unavailability,
after the given set of jobs has already been scheduled to minimize the total weighted completion time.
Such a disruption is represented as an unavailable time interval and is revealed to the production planner before any job is processed;
the production planner wishes to reschedule the jobs to minimize the alteration to the originally planned schedule,
which is measured as the maximum time deviation between the original and the new schedules for all the jobs.
The objective function in this rescheduling problem is to minimize the sum of the total weighted completion time and the weighted maximum time deviation,
under the constraint that the maximum time deviation is bounded above by a given value.
That is, the maximum time deviation is taken both as a constraint and as part of the objective function.
We present a pseudo-polynomial time exact algorithm and a fully polynomial time approximation scheme,
the latter of which is the best possible given that the general problem is NP-hard.
\end{abstract}

\section{Introduction}
In most modern production industries and service systems, various kinds of disruptions will occur,
such as order cancellations, new order arrivals, machine breakdown, and labor or material shortages.
An ideal scheduling system is expected to effectively adjust an originally planned schedule to account for such disruptions,
in order to minimize the effects of the disruption on overall performance.
The extent of an alteration to the originally planned schedule, to be minimized,
becomes either a second objective function (e.g., to model measurable costs),
or is formulated as a constraint to model hard-to-estimate costs,
which may be incorporated into to the original objective function.

In this paper, we investigate the single machine scheduling problem with the objective to minimize the total weighted completion time.
Rescheduling arises because of unexpected machine unavailability,
which we represent as an unavailable time interval. This unavailability
is revealed to the production planner after the given set of jobs has already been scheduled but processing has not begun.
The production planner wishes to reschedule the jobs to minimize the alteration to the originally planned schedule,
measured as the maximum time deviation between the original and the new schedules for all jobs.
The maximum time deviation is taken both as a constraint and as part of the objective function;
that is, the maximum time deviation is bounded above by a given value,
and the new objective function is to minimize the sum of the total weighted completion time and the weighted maximum time deviation.

\subsection{Problem description and definitions}
We formally present our rescheduling problem in what follows, including definitions and notation to be used throughout the paper.

We are given a set of jobs ${\cal J} = \{J_1, J_2, \ldots, J_n\}$,
where the job $J_j$ has an integer weight $w_j$ and requires an integer {\em non-preemptive} processing time $p_j$ on a single machine,
with the original objective to minimize the {\em total weighted completion time}.
This problem is denoted as $(1 \mid \mid \sum_{j=1}^n w_j C_j)$ under the three-field classification scheme~\cite{GLL79},
where $C_j$ denotes the completion time of the job $J_j$.
It is known that the {\em weighted shortest processing time} (WSPT) rule gives an optimal schedule for the problem $(1 \mid \mid \sum_{j=1}^n w_j C_j)$.
We thus assume that the jobs are already sorted in the WSPT order, that is, $\frac {p_1}{w_1} \le \frac {p_2}{w_2} \le \ldots \le \frac {p_n}{w_n}$,
and we denote this order/schedule as $\pi^*$, referred to as the {\em original} schedule
(also called the {\em pre-planned} schedule, or {\em pre-schedule}, in the literature).

The rescheduling arises due to a machine disruption: the machine becomes unavailable in the time interval $[T_1, T_2]$, where $0 \le T_1 < T_2$.
We assume, without loss of generality, that this information is known to us at time zero, so no job is yet processed.
(Otherwise, one may remove those processed jobs from consideration,
and for the partially processed job, either run it to completion and remove it, or stop processing it immediately.)

Let $\sigma$ be a schedule after resolving the machine disruption.
That is, no job of ${\cal J}$ is processed in the time interval $[T_1, T_2]$ in $\sigma$.
As in the existing literature, we use the following notation:
for each job $J_j$,

$S_j(\sigma)$: the starting time of the job $J_j$ in the schedule $\sigma$;

$C_j(\sigma)$: the completion time of the job $J_j$ in the schedule $\sigma$, and thus $C_j(\sigma) = S_j(\sigma) + p_j$;

$C_j(\pi^*)$: the completion time of the job $J_j$ in the original schedule $\pi^*$;

$\Delta_j(\pi^*, \sigma) = |C_j(\sigma) - C_j(\pi^*)|$: the {\em time deviation} of the job $J_j$ in the two schedules.

Let $\Delta_{\max}(\pi^*, \sigma) \triangleq \max_{j=1}^n \{\Delta_j(\pi^*, \sigma)\}$ denote the {\em maximum time deviation} for all jobs.
When it is clear from the context, we simplify the terms $S_j(\sigma)$, $C_j(\sigma)$, $\Delta_j(\pi^*, \sigma)$ and $\Delta_{\max}(\pi^*, \sigma)$
to $S_j$, $C_j$, $\Delta_j$ and $\Delta_{\max}$, respectively.

The time deviation of a job measures how much its actual completion time is off the originally planned,
and thus it can model the penalties resulted from the delivery time change to the satisfaction of the customers,
as well as the cost associated with the rescheduling of resources that are needed before delivery.
In our problem, the time deviation is both taken as a constraint, that is $\Delta_{\max} \le k$ for a given upper bound $k$,
and added to the objective function to minimize $\mu\Delta_{\max} + \sum_{j=1}^n w_j C_j$, for a given balancing factor $\mu \ge 0$.
That is, the goal of rescheduling is to minimize the sum of the weighted maximum time deviation and the total weighted completion time.
Thus, our problem is denoted as $(1, h_1 \mid \Delta_{\max}\le k \mid \mu\Delta_{\max} + \sum_{j=1}^n w_j C_j)$
under the three-field classification scheme~\cite{GLL79},
where the first field ``$1, h_1$'' denotes a single machine with a single unavailable time period,
the second field ``$\Delta_{\max} \le k$'' indicates the constraint on the maximum time deviation,
and the last field is the objective function.

\subsection{Related research}
We next review major research on the variants of the rescheduling problem, inspired by many practical applications.
To name a few such applications,
Bean {\it et al.}~\cite{BBM91} investigated an automobile industry application,
and proposed a heuristic {\em match-up} scheduling approach to accommodate disruptions from multiple sources.
Zweben {\it et al.}~\cite{ZDD93} studied the GERRY scheduling and rescheduling system
that supports Space Shuttle ground processing, using a heuristic constraint-based iterative repair method.
Clausen {\it et al.}~\cite{CHL01} considered a shipyard application,
where the goal for rescheduling is to store large steel plates for efficient access by two non-crossing portal cranes that move the plates to appropriate places.
Vieira {\it et al.}~\cite{VHL03}, Aytug {\it et al.}~\cite{ALM05}, Herroelen and Leus~\cite{HL05} and Yang {\it et al.}~\cite{YQY05}
provided extensive reviews of the rescheduling literature, including taxonomies, strategies and algorithms, for both deterministic and stochastic environments.

In a seminal paper on rescheduling theory for a single machine,
Hall and Potts~\cite{HP04} considered the rescheduling problem required to deal with the arrival of a new set of jobs,
which disrupts the pre-planned schedule of the original jobs.
Such a problem is motivated by the unexpected arrival of new orders in practical manufacturing systems.
First, the set of original jobs has been optimally scheduled to minimize a cost function, typically the maximum lateness or the total completion time; 
but no job has yet been executed.  In this case,
promises have been made to the customers based on the schedule.
Then an unexpected new set of jobs arrives before the processing starts;
the production planner needs to insert the new jobs into the existing schedule seeking to minimize change to the original plan.
The measure of change to the original schedule is the maximum or total sequence deviation,
or the maximum or total time deviation.
For both cases --- where the measure of change is modeled only as a constraint, or
where the measure of change is modeled both as a constraint and is added to the original cost objective --- 
the authors provide either an efficient algorithm or an intractability proof for several problem variants.

Yuan and Mu~\cite{YM07} studied a rescheduling problem similar to the one in \cite{HP04},
but with the objective to minimize the  makespan subject to a limit on the maximum sequence deviation of the original jobs;
they show that such a solution is polynomial time solvable.
Hall {\it et al.}~\cite{HLP07} considered an extension of the rescheduling problem in \cite{HP04},
where the arrivals of multiple new sets of jobs create repeated disruptions to minimize the maximum lateness of the jobs,
subject to a limit on the maximum time deviation of the original jobs;
they proved the NP-hardness and presented several approximation algorithms with their worst-case performance analysis.

Hall and Potts~\cite{HP10} also studied the case where the disruption is a delayed subset of jobs (or called {\em job unavailability}),
with the objective to minimize the total weighted completion time, under a limit on the maximum time deviation;
they presented an exact algorithm, an intractability proof, a constant-ratio approximation algorithm, and a fully polynomial-time approximation scheme (FPTAS).
Hoogeveen {\it et al.}~\cite{HLT12} studied the case where the disruption is the arrival of new jobs and
the machine needs a setup time to switch between processing an original job and processing a new job;
their bi-criterion objective is to minimize the makespan and to minimize the maximum (or total) positional deviation or the maximum (or total) time deviation,
with certain assumptions on the setup times.
They presented a number of polynomial time exact algorithms and intractability proofs for several problem variants.
Zhao and Yuan~\cite{ZY13} examined the case where the disruption is the arrival of new jobs which are associated with release dates,
formulated a bi-objective function to minimize the makespan and to minimize the total sequence deviation, under a limit on the total sequence deviation;
they presented a strongly polynomial-time algorithm for finding all Pareto optimal points of the problem.

Liu and Ro~\cite{LR14} considered the same machine disruption as ours --- the machine is unavailable for a period of time ---
but with the objective to minimize the makespan (or the maximum job lateness) and the weighted maximum time deviation,
under a limit on the maximum time deviation;
they presented a pseudo-polynomial time exact algorithm, a $2$-approximation algorithm, and an FPTAS.
Yin {\it et al.}~\cite{YCW16} studies the rescheduling problem on multiple identical parallel machines with multiple machine disruptions,
and the bi-criterion objective is to minimize the total completion time and to minimize the total virtual tardiness (or the maximum time deviation);
in addition to hardness results, they presented a two-dimensional FPTAS when there is exactly one machine disruption.

Among all related research in the above, the work of Liu and Ro~\cite{LR14} is the most relevant to our work in terms of the scheduling environment,
and the work of Hall and Potts~\cite{HP10} is the most relevant in terms of the original objective function.

\subsection{Our contributions and organization}
Our problem $(1, h_1 \mid \Delta_{\max}\le k \mid \mu\Delta_{\max} + \sum_{j=1}^n w_j C_j)$ includes three interesting special cases,
some of which have received attention in the literature:
when the given bound $k$ is sufficiently large,
the time deviation constraint becomes void and our problem reduces to the total cost problem $(1, h_1 \mid \mid \mu\Delta_{\max} + \sum_{j=1}^n w_j C_j)$;
when the time deviation factor $\mu = 0$,
our problem reduces to the constrained rescheduling problem, without the need to minimize the time deviation;
and finally, when the given bound $k$ is sufficiently large and the time deviation factor $\mu = 0$,
our problem reduces to the classic scheduling problem with a machine unavailability period $(1, h_1 \mid \mid \sum_{j=1}^n w_jC_j)$~\cite{Lee96},
which is NP-hard.

The rest of the paper is organized as follows:
In Section~2, we derive structural properties that are associated with the optimal solutions to our rescheduling problem
$(1, h_1 \mid \Delta_{\max}\le k \mid \mu\Delta_{\max} + \sum_{j=1}^n w_j C_j)$.
In Section~3, we present a pseudo-polynomial time exact algorithm.
In Section~4, we first develop another slower pseudo-polynomial time exact algorithm solving the special case where $\mu = 0$,
this is, the problem $(1, h_1 \mid \Delta_{\max}\le k \mid \sum_{j=1}^n w_j C_j)$.
Based on this slower exact algorithm, we present an FPTAS for our general rescheduling problem.
The FPTAS is an integration of a linear number of FPTASes.
We conclude our paper in the last section, with some final remarks.

\section{Preliminaries}
Firstly, from the NP-hardness of the classic scheduling problem with a machine unavailability period $(1, h_1 \mid \mid \sum_{j=1}^n w_jC_j)$~\cite{Lee96},
we conclude that our rescheduling problem $(1, h_1 \mid \Delta_{\max}\le k \mid \mu\Delta_{\max} + \sum_{j=1}^n w_j C_j)$ is also NP-hard.

Recall that we are given a set of jobs ${\cal J} = \{J_1, J_2, \ldots, J_n\}$,
where each job $J_j$ has a positive weight $w_j$ and a positive processing time $p_j$,
in the WSPT order, a machine unavailability period $[T_1, T_2]$ with $0 \le T_1 < T_2$, an upper bound $k$ on the maximum time deviation,
and a balancing factor $\mu \ge 0$.
All these $w_j$'s, $p_j$'s, $T_1$, $T_2$, $k$ are integers, and $\mu$ is a rational number.
For any feasible schedule $\sigma$ to the rescheduling problem,
from $\Delta_{\max} \le k$ we conclude that for every job $J_j$, $C_j(\pi^*) - k \le C_j(\sigma) \le C_j(\pi^*) + k$.
For ease of presentation, we partition $\sigma$ into two halves, similar to the existing literature:
the prefix of the schedule $\sigma$ with the jobs completed before or at time $T_1$ is referred to as the {\em earlier schedule} of $\sigma$,
and the suffix of the schedule $\sigma$ with the jobs completed after time $T_2$ is referred to as the {\em later schedule} of $\sigma$.
We assume, without loss of generality, that with the same $\Delta_{\max}$, all the jobs are processed as early as possible in $\sigma$
(to achieve the minimum possible total weighted completion time).
Let $\sigma^*$ denote an optimal schedule to the rescheduling problem.

\subsection{Problem setting}
Let
\begin{equation}
\label{eq1}
p_{\min} \triangleq \min_{j=1}^n p_j, \ p_{\max} \triangleq \max_{j=1}^n p_j, \ P \triangleq \sum_{j=1}^n p_j,
\ w_{\max} \triangleq \max_{j=1}^n w_j, \mbox{ and } W \triangleq \sum_{j=1}^n w_j.
\end{equation}

Using the original schedule $\pi^* = (1, 2, \ldots, n)$, we compute
\begin{equation}
\label{eq2}
j_1 \triangleq \min\{j \mid C_j(\pi^*) > T_1\}, \quad j_2 \triangleq \min\{j \mid S_j(\pi^*) \ge T_2\},
\end{equation}
i.e., $J_{j_1}$ is the first job in $\pi^*$ completed strictly after time $T_1$ and $J_{j_2}$ is the first job in $\pi^*$ starting processing at or after time $T_2$.
One clearly sees that if $j_1$ is void, then no rescheduling is necessary;  in the sequel, we always assume that the job $J_{j_1}$ exists.
Nevertheless we note that $j_2$ could be void, which means that all the jobs start processing strictly before time $T_2$ in $\pi^*$.

We may furthermore assume the following relations hold among $p_{\min}, P, T_1, T_2$ and $k$, to ensure that the rescheduling problem is non-trivial:
\begin{equation}
\label{eq3}
p_{\min} \le T_1 < P, \mbox{ and } T_2 - S_{j_1}(\pi^*) \le k. 
\end{equation}
%
%
For a quick proof, firstly, if $T_1 < p_{\min}$, then no job can be processed before the machine unavailability period and
thus the schedule $\pi^*$ remains optimal
except that the job processing starts at time $T_2$ instead of time $0$;
secondly, if $T_1 \ge P$, then no rescheduling is necessary.
Lastly, from the definition of the job $J_{j_1}$ in Eq.~(\ref{eq2}),
we conclude that at least one job among $J_1, J_2, \ldots, J_{j_1}$ must be completed after time $T_2$ in any feasible rescheduling solution,
with its time deviation at least $T_2 - S_{j_1}(\pi^*)$.
Therefore, $T_2 - S_{j_1}(\pi^*) \le k$, as otherwise no feasible solution exists.


\subsection{Structure properties of the optimal schedules}\label{sp}
There is a very regular property of our target optimal schedules to the rescheduling problem, stated in the following lemma.

\begin{lemma}
\label{lemma01}
There exists an optimal schedule $\sigma^*$ for the rescheduling problem $(1, h_1 \mid \Delta_{\max}\le k \mid \mu\Delta_{\max} + \sum_{j=1}^n w_j C_j)$,
in which
\begin{description}
\parskip=0pt
\item[(a)]
	the jobs in the earlier schedule are in the same order as they appear in $\pi^*$;
\item[(b)]
	the jobs in the later schedule are also in the same order as they appear in $\pi^*$.
\end{description}
\end{lemma}
\begin{proof}
By contradiction, assume $(J_i, J_j)$ is the first pair of jobs for which $J_i$ precedes $J_j$ in $\pi^*$, i.e. $p_i/w_i \le p_j/w_j$,
but $J_j$ immediately precedes $J_i$ in the earlier schedule of $\sigma^*$.
Let $\sigma'$ denote the new schedule obtained from $\sigma^*$ by swapping $J_j$ and $J_i$.
If $C_j(\sigma') \ge C_j(\pi^*)$, then $C_i(\sigma') \ge C_i(\pi^*)$ too,
and thus $\Delta_j(\sigma', \pi^*) \le \Delta_i(\sigma', \pi^*) < \Delta_i(\sigma^*, \pi^*)$ due to $p_j > 0$;
if $C_j(\sigma') < C_j(\pi^*)$ and $C_i(\sigma') \le C_i(\pi^*)$,
then $\Delta_i(\sigma', \pi^*) \le \Delta_j(\sigma', \pi^*) < \Delta_j(\sigma^*, \pi^*)$ due to $p_i > 0$;
lastly if $C_j(\sigma') < C_j(\pi^*)$ and $C_i(\sigma') > C_i(\pi^*)$,
then $\Delta_i(\sigma', \pi^*) < \Delta_i(\sigma^*, \pi^*)$ and $\Delta_j(\sigma', \pi^*) < \Delta_j(\sigma^*, \pi^*)$.
That is, $\sigma'$ is also a feasible reschedule.

Furthermore, the weighted completion times contributed by $J_i$ and $J_j$ in $\sigma'$ is no more than those in $\sigma^*$,
implying the optimality of $\sigma'$.
It follows that, if necessary,
after a sequence of job swappings, we will obtain an optimal reschedule in which the jobs in the earlier schedule are in the same order as they appear in $\pi^*$.

The second part of the lemma can be similarly proved (see the Appendix A).
\end{proof}

There are several more properties in the following Lemma~\ref{lemma02}, which are important to the design and analysis of the algorithms to be presented.
Most of these properties also hold for the optimal schedules to a similar makespan rescheduling problem
$(1, h_1 \mid \Delta_{\max}\le k \mid \mu\Delta_{\max} + C_{\max})$~\cite{LR14}.
We remark that the makespan is only a part of our objective,
and the original schedule for the makespan scheduling problem $(1, h_1 \mid \Delta_{\max}\le k \mid \mu\Delta_{\max} + C_{\max})$ can be arbitrary.
However, for our problem, the jobs in $\pi^*$ are in the special WSPT order.
Our goal is to compute an optimal schedule satisfying (the properties stated in) Lemmas~\ref{lemma01} and \ref{lemma02},
and thus we examine only those feasible schedules $\sigma$ satisfying Lemmas~\ref{lemma01} and \ref{lemma02}.

Recall that we assume, for every feasible schedule $\sigma$ to the rescheduling problem
with the same $\Delta_{\max}$, all the jobs are processed as early as possible in $\sigma$
(to achieve the minimum possible total weighted completion time).
However this does not rule out the possibility that the machine would idle.
In fact, the machine has to idle for a period of time right before time $T_1$, if no job can be fitted into this period for processing;
also, the machine might choose to idle for a period of time and then proceeds to process a job, in order to obtain a smaller $\Delta_{\max}$.
In the sequel, our discussion is about the latter kind of machine idling.
The good news is that there are optimal schedules in which the machine has at most one such idle time period, as shown in the following Lemma~\ref{lemma02}.
This makes our search for optimal schedules much easier.

\begin{lemma}
\label{lemma02}
There exists an optimal schedule $\sigma^*$ for the rescheduling problem $(1, h_1 \mid \Delta_{\max}\le k \mid \mu\Delta_{\max} + \sum_{j=1}^n w_j C_j)$,
in which
\begin{description}
\parskip=0pt
\item[(a)]
	$C_j(\sigma^*) \le C_j(\pi^*)$ for each job $J_j$ in the earlier schedule;
\item[(b)]
	the machine idles for at most one period of time in the earlier schedule;
\item[(c)]
	each job in the earlier schedule after the idle time period is processed exactly $\Delta_{\max}$ time units earlier than in $\pi^*$;
\item[(d)]
	the jobs in the earlier schedule after the idle time period are consecutive in $\pi^*$;
\item[(e)]
	the job in the earlier schedule right after the idle time period has a starting time at or after time $T_2$ in $\pi^*$;
\item[(f)]
	the machine does not idle in the later schedule;
\item[(g)]
	the first job in the later schedule reaches the maximum time deviation among all the jobs in the later schedule.
\end{description}
\end{lemma}
\begin{proof}
Item (a) is a direct consequence of Lemma~\ref{lemma01}, since the jobs in the earlier schedule are in the same order as they appear in $\pi^*$,
which is the WSPT order.

Proofs of (b)--(g) are similar to those in \cite{LR14}, where they are proven for an arbitrary original schedule, while the WSPT order is only a special order.
For completeness, the proofs are included in the Appendix B.
\end{proof}

Among the jobs $J_1, J_2, \ldots, J_{j_1}$, we know that some of them will be processed in the later schedule of the optimal schedule $\sigma^*$.
By Lemma~\ref{lemma01},
we conclude that the first job in the later schedule is from $J_1, J_2, \ldots, J_{j_1}$.
We use $J_a$ to denote this job, and consequently $(J_1, J_2, \ldots, J_{a-1})$ remains as the prefix of the earlier schedule.
The time deviation of $J_a$ is $\Delta_a = T_2 - S_a(\pi^*) \le k$ (which could be used to further narrow down the candidates for $J_a$).

\begin{corollary}
\label{coro03}
In an optimal schedule $\sigma^*$ satisfying Lemmas~\ref{lemma01} and \ref{lemma02},
suppose the job $J_a$ is the first job in the later schedule.
For each $j = a+1, a+2, \ldots, {j_2 - 1}$, if the job $J_j$ is in the earlier schedule, then its time deviation $\Delta_j$ is less than $\Delta_a$.
\end{corollary}
\begin{proof}
From the definition of $j_2$ in Eq.~(\ref{eq2}), we know that $S_j(\pi^*) < T_2$, and therefore its time deviation is $\Delta_j < T_2 - S_a(\pi^*) = \Delta_a$.
\end{proof}

\section{A dynamic programming exact algorithm}
In this section, we develop an exact algorithm {\sc DP-1} for the rescheduling problem $(1, h_1 \mid \Delta_{\max}\le k \mid \mu\Delta_{\max} + \sum_{j=1}^n w_j C_j)$,
to compute an optimal schedule $\sigma^*$ satisfying Lemmas~\ref{lemma01} and \ref{lemma02}.
The key idea is as follows:
By Lemma~\ref{lemma02}(g), we first guess the job $J_a$ that starts processing at time $T_2$ in $\sigma^*$,
with its time deviation $\Delta_a = T_2 - S_a(\pi^*) \le k$.
From this initial partial schedule,%
\footnote{In the entire paper, we examine only feasible partial schedules restricted to the considered jobs;
these partial schedules can always be completed into certain feasible full schedules.}
our algorithm constructs some feasible full schedules satisfying Lemmas~\ref{lemma01} and \ref{lemma02}.
To guarantee that our algorithm constructs an optimal full schedule in pseudo-polynomial time,
we use a ``hash'' function to map a partial schedule into a quadruple, such that only one partial schedule per quadruple is used in the computation,
to be described in detail in the following.
At the end, the full schedule with the minimum objective function value is returned as the solution $\sigma^*$.

We notice from Lemmas~\ref{lemma01} and \ref{lemma02}(a) that the time deviations of the jobs in the earlier schedule of $\sigma^*$ are non-decreasing,
with the first $(a-1)$ ones being $0$'s.
From the constraint $\Delta_{\max} \le k$, the last job in the schedule $\pi^*$ that can possibly be in the earlier schedule of $\sigma^*$ is $J_{j_3}$ with
\begin{equation}
\label{eq4}
j_3 \triangleq \max\{j \mid C_j(\pi^*) - k \le T_1\}.
\end{equation}

We use $P_i$ to denote the total processing time of the first $i$ jobs in $\pi^*$:
\begin{equation}
\label{eq5}
P_i \triangleq \sum_{j=1}^i p_j, \mbox{ for } i = 1, 2, \ldots, n.
\end{equation}

The exact algorithm {\sc DP-1} is a dynamic programming,
that sequentially assigns the job $J_{j+1}$ to the partial schedules on the first $j$ jobs $J_1, J_2, \ldots, J_j$;
each such partial schedule is described (hashed) as a quadruple $(j; \ell_j, e_j; Z_j)$,
where $a \le j \le j_3$,
no machine idling period in the earlier schedule,
$\ell_j$ is the total processing time of the jobs in the earlier schedule,
$e_j$ is the index of the last job in the earlier schedule,
and $Z_j$ is the total weighted completion time of the jobs in the partial schedule.
Note that only one, though arbitrary, partial schedule is saved by {\sc DP-1} for each quadruple.

For ease of presentation, we partition the jobs into four subsequences:
\begin{equation}
\label{eq6}
\begin{array}{lll}
{\cal J}_1	&=	&(J_1, J_2, \ldots, J_{a-1}),\\
{\cal J}_2	&=	&(J_{a+1}, J_{a+2}, \ldots, J_{j_2-1}),\\
{\cal J}_3	&=	&(J_{j_2}, J_{j_2+1}, \ldots, J_{j_3}),\\
{\cal J}_4	&=	&(J_{j_3 + 1}, J_{j_3 + 2}, \ldots, J_n).
\end{array}
\end{equation}
From Lemma~\ref{lemma01} and Eq.~(\ref{eq4}), we know that ${\cal J}_1$ is a prefix of the earlier schedule and ${\cal J}_4$ is a suffix of the later schedule,
and thus {\sc DP-1} takes care of only the jobs of ${\cal J}_2 \cup {\cal J}_3$.
Let $Z({\cal J}_1)$ denote the total weighted completion time of the jobs in ${\cal J}_1$,
and $Z({\cal J}_4)$ denote the total weighted completion time of the jobs in ${\cal J}_4$ by starting the job processing at time $0$.

The (only) starting partial schedule for {\sc DP-1} is described as
\begin{equation}
\label{eq7}
(a; \ell_a, e_a; Z_a) = (a; P_{a-1}, a-1; Z({\cal J}_1) + w_a (T_2 + p_a)),
\end{equation}
in which the job $J_a$ starts processing at time $T_2$.
In general, given a quadruple $(j; \ell_j, e_j; Z_j)$ with $a \le j < j_3$, that represents a partial schedule on the first $j$ jobs,
{\sc DP-1} assigns the next job $J_{j+1}$ of ${\cal J}_2 \cup {\cal J}_3$ as follows
to generate at most three new partial schedules each described as a quadruple $(j+1; \ell_{j+1}, e_{j+1}; Z_{j+1})$.
Furthermore, if a non-empty machine idling period is inserted in the earlier schedule of a new partial schedule,
then the partial schedule is directly completed optimally to a full schedule using Lemma~\ref{lemma02}(d).

{\bf Case 1.} $J_{j+1}$ is added in the later schedule to obtain a partial schedule described as:
\begin{equation}
\label{eq8}
(j+1; \ell_{j+1}, e_{j+1}; Z_{j+1}) = (j+1; \ell_j, e_j; Z_j + w_{j+1} (T_2 + P_{j+1} - \ell_{j})),
\end{equation}
in which $C_{j+1} = T_2 + P_{j+1} - \ell_{j}$.
The feasibility holds since $\Delta_{j+1} \le \Delta_a$ by Lemma~\ref{lemma02}(g).

{\bf Case 2.} If $J_{j+1}$ can fit in the earlier schedule, that is $\ell_j + p_{j+1} \le T_1$ and $C_{j+1}(\pi^*) - (\ell_j + p_{j+1}) \le k$,
then we add $J_{j+1}$ in the earlier schedule without inserting a machine idling period to obtain a feasible partial schedule described as:
\begin{equation}
\label{eq9}
(j+1; \ell_{j+1}, e_{j+1}; Z_{j+1}) = (j+1; \ell_j + p_{j+1}, j+1; Z_j + w_{j+1} (\ell_j + p_{j+1})),
\end{equation}
in which $C_{j+1} = \ell_j + p_{j+1}$ and $\Delta_{j+1} = C_{j+1}(\pi^*) - (\ell_j + p_{j+1})$.
From Corollary~\ref{coro03}, for $J_{j+1} \in {\cal J}_2$ we have $\Delta_{j+1} < \Delta_a$ and thus no need for checking the second inequality.

{\bf Case 3.} If $\ell_j + p_{j+1} < T_1$ and $C_{j+1}(\pi^*) - (\ell_j + p_{j+1}) > \max\{\Delta_{e_j}, \Delta_a\}$,
then we add $J_{j+1}$ in the earlier schedule and insert a non-empty machine idling period right before it.
The following Lemma~\ref{lemma04} states that the exact length of the machine idling period can be determined in $O(n-j)$-time,
and during the same time, the partial schedule is directly optimally completed into a full schedule using Lemma~\ref{lemma02}(d).

Lastly, for every quadruple $(j; \ell_j, e_j; Z_j)$ with $j = j_3$ representing a partial schedule on the first $j_3$ jobs,
the algorithm {\sc DP-1} completes it by assigning all the jobs of ${\cal J}_4$ in the later schedule, starting at time $T_2 + P_j - \ell_j$,
to obtain a full schedule described as the quadruple
\begin{equation}
\label{eq10}
(n; \ell_n, e_n; Z_n) =
\left(n; \ell_j, e_j; Z_j + Z({\cal J}_4) + \left(\sum_{i = j+1}^n w_i\right) (T_2 + P_j - \ell_j)\right),
\end{equation}
for which the maximum time deviation is $\Delta_{\max} = \max\{\Delta_{e_n}, \Delta_a\}$ and the objective function value is
\begin{equation}
\label{eq11}
\hat{Z}_n \triangleq \mu\max\{\Delta_{e_n}, \Delta_a\} + Z_n.
\end{equation}

\begin{lemma}
\label{lemma04}
When the job $J_{j+1}$ is added in the earlier schedule of a partial schedule described as $(j; \ell_j, e_j; Z_j)$
and a non-empty machine idling period is inserted before it, then
\begin{description}
\parskip=0pt
\item[(a)]
	the minimum possible length of this period is $\max\{1, C_{j+1}(\pi^*) - (\ell_j + p_{j+1}) - k\}$;
\item[(b)]
	the maximum possible length of this period is $\min\{C_{j+1}(\pi^*) - (\ell_j + p_{j+1}) - \max\{\Delta_{e_j}, \Delta_a\}$,
	$T_1 - (\ell_j + p_{j+1})\}$;
\item[(c)]
	the exact length can be determined in $O(n-j)$-time;
\item[(d)]
	during the same time an optimal constrained full schedule is achieved directly.
\end{description}
\end{lemma}
\begin{proof}
$1$ is a trivial lower bound given that all job processing times are positive integers.
Next, if $C_{j+1}(\pi^*) - (\ell_j + p_{j+1}) > k$,
then we cannot process $J_{j+1}$ immediately at time $\ell_j$ as this would result in a time deviation greater than the given upper bound $k$;
and the earliest possible starting time is $C_{j+1}(\pi^*) - p_{j+1} - k$, leaving a machine idling period of length $C_{j+1}(\pi^*) - (\ell_j + p_{j+1}) - k$.
This proves item (a).

The maximum possible length for the period is no more than $C_{j+1}(\pi^*) - (\ell_j + p_{j+1}) - \max\{\Delta_{e_j}, \Delta_a\}$,
for otherwise $\Delta_{j+1}$ would be less than $\max\{\Delta_{e_j}, \Delta_a\}$.
However, $\Delta_{j+1} < \Delta_{e_j}$ contradicts Lemma~\ref{lemma01} and Lemma~\ref{lemma02}(a)
that suggest the jobs in the earlier schedule should have non-decreasing time deviations;
$\Delta_{j+1} < \Delta_a$ contradicts Lemma~\ref{lemma02}(c) that $\Delta_{\max} \ge \Delta_a$.
Also, since $J_{j+1}$ is added to the earlier schedule, the length of the idling period has to be at most $T_1 - (\ell_j + p_{j+1})$.
This proves item (b).

Let $L = \max\{1, C_{j+1}(\pi^*) - (\ell_j + p_{j+1}) - k\}$ and
$U = \min\{C_{j+1}(\pi^*) - (\ell_j + p_{j+1}) - \max\{\Delta_{e_j}, \Delta_a\}, T_1 - (\ell_j + p_{j+1})\}$. 
For each value $i \in [L, U]$, we
1) start processing $J_{j+1}$ at time $\ell_j + i$,
2) then continuously process succeeding jobs in the earlier schedule until the one won't fit in,
and 3) lastly process all the remaining jobs in the later schedule.
This gives a full schedule, denoted as $\pi^i$, with the total weighted completion time denoted $Z^i_n$,
and the maximum time deviation $\Delta^i_{\max} = \Delta_{j+1} = C_{j+1}(\pi^*) - (\ell_j + p_{j+1}) - i$.
Its objective function value is $\hat{Z}^i_n \triangleq \mu\Delta^i_{\max} + Z^i_n$.

It follows that the interval $[L, U]$ can be partitioned into $O(n-j)$ sub-intervals,
such that for all values in a sub-interval say $[i_1, i_2]$, the jobs assigned to the earlier schedule are identical;
and consequently the objective function value $\hat{Z}^i_n$ is a linear function in $i$, where $i_1 \le i \le i_2$.
It follows that among all these full schedules, the minimum objective function value must be achieved at one of $\hat{Z}^{i_1}_n$ and $\hat{Z}^{i_2}_n$.
That is, we in fact do not need to compute the full schedules $\pi^i$'s with those $i$'s such that $i_1 < i < i_2$,
and consequently there are only $O(n-j)$ full schedules to be computed.

These sub-intervals can be determined as follows:
when $i = L$, let the jobs fit into the earlier schedule after the machine idling period be $J_{j+1}, J_{j+2}, \ldots, J_{j+s}$,
then the first sub-interval is $[L, L + T_1 - C_{j+s}]$, where $C_{j+s}$ is the completion time of $J_{j+s}$ in the full schedule $\pi^i$;
the second sub-interval is $[L + T_1 - C_{j+s} + 1, L + T_1 - C_{j+s-1}]$;
the third sub-interval is $[L + T_1 - C_{j+s-1} + 1, L + T_1 - C_{j+s-2}]$;
and so on,
until the last interval hits the upper bound $U$.

The optimal length of the machine idling period is the one $i^*$ that minimizes $\hat{Z}^i_n$, among the $O(n-j)$ computed full schedules,
and we obtain a corresponding constrained optimal full schedule directly from the partial schedule described as the quadruple $(j; \ell_j, e_j; Z_j)$.
\end{proof}

\begin{theorem}
\label{thm05}
The algorithm {\sc DP-1} solves the rescheduling problem $(1, h_1 \mid \Delta_{\max}\le k \mid \mu\Delta_{\max} + \sum_{j=1}^n w_j C_j)$,
with its running time in $O(n^3 T_1 w_{\max} (T_2 + P))$,
where $n$ is the number of jobs, $[T_1, T_2]$ is the machine unavailable time interval, $P$ is the total job processing time,
and $w_{\max}$ is the maximum weight of the jobs.
\end{theorem}
\begin{proof}
In the above, we presented the dynamic programming algorithm {\sc DP-1} to compute a full schedule for each quadruple $(n; \ell_n, e_n; Z_n)$
satisfying Lemmas~\ref{lemma01} and \ref{lemma02},
under the constraint that the job $J_a$ starts processing at time $T_2$.
The final output schedule is the one with the minimum $\hat{Z}_n$, among all the choices of $J_a$ such that $\Delta_a \le k$.
Its optimality lies in the quadruple representation for the partial schedules and the recurrences we developed in Eqs.~(\ref{eq7}--\ref{eq10}).
There are $O(j_1)$ choices for $J_a$, and for each $J_a$ we compute a partial schedule for each quadruple $(j; \ell_j, e_j; Z_j)$,
where $a \le j \le j_3$ or $j = n$, $P_{a-1} \le \ell_j \le T_1$, $a-1 \le e_j \le j$, and $Z_j \le \max_j w_j (T_2 + P)$.
Since one partial schedule leads to at most three other partial schedules, and each takes an $O(1)$-time to compute,
the overall running time is $O(n^3 T_1 w_{\max} (T_2 + P))$.
\end{proof}

\section{An FPTAS}
The route to the FPTAS for the rescheduling problem $(1, h_1 \mid \Delta_{\max}\le k \mid \mu\Delta_{\max} + \sum_{j=1}^n w_j C_j)$ is as follows:
we first consider the special case where $\mu = 0$, that is $(1, h_1 \mid \Delta_{\max}\le k \mid \sum_{j=1}^n w_j C_j)$,
where the maximum time deviation is only upper bounded by $k$ but not taken as part of the objective function,
and design another exact algorithm, denoted as {\sc DP-2}, using a different dynamic programming recurrence;
then we develop from the algorithm {\sc DP-2} an FPTAS for the special case;
lastly, we use the FPTAS for the special case multiple times to design an FPTAS for the general case.

\subsection{Another dynamic programming exact algorithm for $\mu = 0$}
In this special case, we have stronger conclusions on the target optimal schedule $\sigma^*$ than those stated in Lemma \ref{lemma02},
one of which is that if there is a machine idling period in the earlier schedule of $\sigma^*$, then $\Delta_{\max} = k$.
This follows from the fact that we may start processing the jobs after the idling period one time unit earlier, if $\Delta_{\max} < k$,
to decrease the total weighted completion time.
We conclude this in the following lemma.

\begin{lemma}
\label{lemma06}
There exists an optimal schedule $\sigma^*$ for the rescheduling problem $(1, h_1 \mid \Delta_{\max}\le k \mid \sum_{j=1}^n w_j C_j)$,
in which if the machine idles in the earlier schedule then $\Delta_{\max} = k$.
\end{lemma}

The new exact algorithm {\sc DP-2} heavily relies on this conclusion (and thus it does not work for the general case).
{\sc DP-2} has a running time complexity worse than {\sc DP-1},
but it can be readily developed into an FPTAS.

The framework of the new algorithm {\sc DP-2} is the same, and we continue to use the notations defined in Eqs.~(\ref{eq4}--\ref{eq6}).
But now we use a new quadruple $(j; \ell_j^1, \ell_j^2; Z_j)$ to describe a partial schedule on the first $j$ jobs,
where $a \le j \le j_3$,
no machine idling period in the earlier schedule,
$\ell_j^1$ is the maximum job completion time in the earlier schedule,
$\ell_j^2$ is the maximum job completion time in the later schedule,
and $Z_j$ is the total weighted completion time of the jobs in the partial schedule.

Starting with guessing $J_a$ (from the pool $J_1, J_2, \ldots, J_{j_1}$) to be the job started processing at time $T_2$,
such that $\Delta_a = T_2 - S_a(\pi^*) \le k$,
the (only) corresponding partial schedule is described as the quadruple
\begin{equation}
\label{eq12}
(a; \ell_a^1, \ell_a^2; Z_a) = (a; P_{a-1}, T_2 + p_a; Z({\cal J}_1) + w_a (T_2 + p_a)).
\end{equation}
In general, given a quadruple $(j; \ell_j^1, \ell_j^2; Z_j)$ with $a \le j < j_3$ representing a partial schedule on the first $j$ jobs,
the algorithm {\sc DP-2} assigns the next job $J_{j+1}$ of ${\cal J}_2 \cup {\cal J}_3$ as follows
to generate at most three new partial schedules each described as a quadruple $(j+1; \ell_{j+1}^1, \ell_{j+1}^2; Z_{j+1})$.
Furthermore, if a non-empty machine idling period is inserted in the earlier schedule of a new partial schedule,
then it is directly completed optimally into a full schedule using Lemma~\ref{lemma02}(d).

{\bf Case 1.} $J_{j+1}$ is added in the later schedule to obtain a partial schedule described as:
\begin{equation}
\label{eq13}
(j+1; \ell_{j+1}^1, \ell_{j+1}^2; Z_{j+1}) = (j+1; \ell_j^1, \ell_j^2 + p_{j+1}; Z_j + w_{j+1} (\ell_j^2 + p_{j+1})),
\end{equation}
in which $C_{j+1} = \ell_j^2 + p_{j+1}$.
The feasibility holds since $\Delta_{j+1} \le \Delta_a$ by Lemma~\ref{lemma02}(g).

{\bf Case 2.} If $J_{j+1}$ can fit in the earlier schedule, that is $\ell_j^1 + p_{j+1} \le T_1$ and $C_{j+1}(\pi^*) - (\ell_j^1 + p_{j+1}) \le k$,
then we add $J_{j+1}$ in the earlier schedule without inserting a machine idling period to obtain a feasible partial schedule described as:
\begin{equation}
\label{eq14}
(j+1; \ell_{j+1}^1, \ell_{j+1}^2; Z_{j+1}) = (j+1; \ell_j^1 + p_{j+1}, \ell_j^2; Z_j + w_{j+1} (\ell_j^1 + p_{j+1})),
\end{equation}
in which $C_{j+1} = \ell_j^1 + p_{j+1}$ and $\Delta_{j+1} = C_{j+1}(\pi^*) - (\ell_j^1 + p_{j+1})$.
From Corollary~\ref{coro03}, for $J_{j+1} \in {\cal J}_2$ we have $\Delta_{j+1} < \Delta_a$ and thus no need for checking the second inequality.

{\bf Case 3.} If $\ell_j^1 + p_{j+1} < C_{j+1}(\pi^*) - k \le T_1$,
then we add $J_{j+1}$ in the earlier schedule and insert a non-empty machine idling period of length $C_{j+1}(\pi^*) - (\ell_j^1 + p_{j+1}) - k$ right before it.
(That is, start processing $J_{j+1}$ at time $C_{j+1}(\pi^*) - p_{j+1} - k$.)
Then continuously process succeeding jobs up to $J_{j_3}$ in the earlier schedule,
and lastly process all the remaining jobs in the later schedule.
This gives a full schedule with the maximum time deviation $\Delta_{\max} = k$, described as the quadruple $(n; \ell_n^1, \ell_n^2; Z_n)$, where
\[
\ell_n^1 = C_{j+1}(\pi^*) - k + \sum_{i = j+1}^{j_3} p_i, \quad
\ell_n^2 = \ell_j^2 + \sum_{i = j_3+1}^n p_i, \mbox{ and}
\]
\[
Z_n = Z_j + Z(\{J_{j+1}, \ldots, J_{j_3}\}) + \left(\sum_{i = j+1}^{j_3} w_i\right) (C_{j+1}(\pi^*) - p_{j+1} - k)
	+ Z({\cal J}_4) + \left(\sum_{i = j+1}^n w_i\right) \ell_j^2,
\]
where $Z(\{J_{j+1}, \ldots, J_{j_3}\})$ denotes the total weighted completion time
of the jobs in $\{J_{j+1}, \ldots$, $J_{j_3}\}$ by starting the job processing at time $0$.

Lastly, for every quadruple $(j; \ell_j^1, \ell_j^2; Z_j)$ with $j = j_3$ representing a partial schedule on the first $j_3$ jobs,
the algorithm {\sc DP-2} completes it by assigning all the jobs of ${\cal J}_4$ in the later schedule, starting at time $\ell_j^2$,
to obtain a full schedule described as the quadruple
\begin{equation}
\label{eq15}
(n; \ell_n^1, \ell_n^2; Z_n) =
\left(n; \ell_j^1, \ell_j^2 + \sum_{i = j+1}^n p_i; Z_j + Z({\cal J}_4) + \left(\sum_{i = j+1}^n w_i\right) \ell_j^2\right),
\end{equation}
for which the maximum time deviation is guaranteed to be no more than $k$.

\begin{theorem}
\label{thm07}
The algorithm {\sc DP-2} solves the rescheduling problem $(1, h_1 \mid \Delta_{\max}\le k \mid \sum_{j=1}^n w_j C_j)$ in $O(n^2 T_1 w_{\max} (T_2 + P)^2)$-time,
where $n$ is the number of jobs, $[T_1, T_2]$ is the machine unavailable time interval, $P$ is the total job processing time,
and $w_{\max}$ is the maximum weight of the jobs.
\end{theorem}
\begin{proof}
In the above, we see that the dynamic programming algorithm {\sc DP-2} computes a full schedule for each quadruple $(n; \ell_n^1, \ell_n^2; Z_n)$
satisfying Lemmas~\ref{lemma01} and \ref{lemma02}
and an extra property that if there is a machine idling period in the earlier schedule then $\Delta_{\max} = k$,
under the constraint that the job $J_a$ starts processing at time $T_2$.
The final output schedule is the one with the minimum $Z_n$, among all the choices of $J_a$ such that $\Delta_a \le k$.
Similarly to the proof of Theorem~\ref{thm05},
the optimality lies in the quadruple representation for the partial schedules and the recurrences we developed in Eqs.~(\ref{eq12}--\ref{eq15}).
There are $O(j_1)$ choices for $J_a$, and for each $J_a$ we compute all partial schedules described as $(j; \ell_j^1, \ell_j^2; Z_j)$,
where $a \le j \le n$, $P_{a-1} \le \ell_j^1 \le T_1$, $T_2 + p_a \le \ell_j^2 \le T_2 + P$, and $Z_j \le \max_j w_j (T_2 + P)$.
Since one partial schedule leads to at most three other partial schedules, and each takes an $O(1)$-time to compute,
the overall running time is $O(n^2 T_1 w_{\max} (T_2 + P)^2)$.
\end{proof}

We remark that if the purpose is to compute an optimal schedule only,
the algorithm {\sc DP-2} can be an order faster than the exact algorithm {\sc DP-1} in Theorem~\ref{thm05},
by dropping the third entry $\ell_j^2$ to run in $O(n^2 T_1 w_{\max} (T_2 + P))$-time.

\subsection{An FPTAS for $\mu = 0$}
In this subsection, we convert the exact algorithm {\sc DP-2} in the last subsection into an FPTAS by the sparsing technique.
We assume the job $J_a$ is scheduled to start processing at time $T_2$, and we have a positive real value $\epsilon > 0$.
The algorithm is denoted as {\sc Approx$(a, \epsilon)$}, which guarantees to return a feasible schedule such that
its total weighted completion time is within $(1 + \epsilon)$ of the constrained optimum,
under the constraint that the job $J_a$ is scheduled to start processing at time $T_2$.

\ 

\noindent{\bf Step 1.} \ 
Set
\begin{equation}
\label{eq16}
\delta \triangleq (1 + \epsilon /2n).
\end{equation}
Note that a partial schedule is described as the quadruple $(j; \ell_j^1, \ell_j^2; Z_j)$,
where $a \le j \le n$,
$\ell_j^1 = 0$ or $p_{\min} \le \ell_j^1 \le T_1$,
$T_2 + p_a \le \ell_j^2 \le T_2 + P$,
and $Z({\cal J}_1) + w_a (T_2 + p_a) \le Z_j \le Z({\cal J}) + W T_2$,
where $W \triangleq \sum_j w_j$ is the total weight.
Let $L^1, U^1$ ($L^2, U^2$; $L^3, U^3$, respectively) denote the above lower and upper bounds on $\ell_j^1$ ($\ell_j^2$; $Z_j$, respectively);
let
\[
r_i \triangleq \lfloor \log_{\delta} (U^i /L^i)\rfloor, \mbox{ for } i = 1, 2, 3,
\]
and split the interval $[L^i, U^i]$ into subintervals
\[
I_1^i = [L^i, L^i \delta], \ I_2^i = (L^i \delta, L^i \delta^2], \ \ldots, \ I_{r_i}^i = (L^i \delta^{r_i}, U^i], \mbox{ for } i = 1, 2, 3.
\]

We define a three-dimensional box
\begin{equation}
\label{eq17}
B_{i_1, i_2, i_3} \triangleq I_{i_1}^1 \times I_{i_2}^2 \times I_{i_3}^3, \mbox{ for } (i_1, i_2, i_3) \in (\{0\} \cup [r_1]) \times [r_2] \times [r_3],
\end{equation}
where $I_0^1 \triangleq [0, 0]$ and $[r] \triangleq \{1, 2, \ldots, r\}$ for any positive integer $r$.
Each $j$ such that $a \le j \le j_3$ or $j = n$ is associated with a box $B_{i_1, i_2, i_3}$, denoted as $j$-$B_{i_1, i_2, i_3}$ for simplicity,
where $(i_1, i_2, i_3) \in (\{0\} \cup [r_1]) \times [r_2] \times [r_3]$;
all these boxes are initialized empty.

\ 

\noindent{\bf Step 2.} \ 
The starting partial schedule is described as $(a; \ell_a^1, \ell_a^2; Z_a)$ in Eq.~(\ref{eq12}),
which is then saved in the box $a$-$B_{i_1, i_2, i_3}$ where $\ell_a^1 \in I_{i_1}^1$, $\ell_a^2 \in I_{i_2}^2$, and $Z_a \in I_{i_3}^3$.

In general, for each non-empty box $j$-$B_{i_1, i_2, i_3}$ with $a \le j < j_3$,
denote the saved quadruple as $(j; \ell_j^1, \ell_j^2; Z_j)$, which describes a partial schedule on the first $j$ jobs such that
there is no machine idling period in the earlier schedule,
$\ell_j^1$ is the maximum job completion time in the earlier schedule,
$\ell_j^2$ is the maximum job completion time in the later schedule,
and $Z_j$ is the total weighted completion time of the jobs in the partial schedule.
The algorithm {\sc Approx$(a, \epsilon)$} performs the same as the exact algorithm {\sc DP-2} to assign the next job $J_{j+1}$ of ${\cal J}_2 \cup {\cal J}_3$
to generate at most three new partial schedules each described as a quadruple $(j+1; \ell_{j+1}^1, \ell_{j+1}^2; Z_{j+1})$.
Furthermore, if a non-empty machine idling period is inserted in the earlier schedule of a new partial schedule,
then it is directly completed optimally into a full schedule $(n; \ell_n^1, \ell_n^2; Z_n)$ using Lemma~\ref{lemma02}(d).
For each resultant $(j+1; \ell_{j+1}^1, \ell_{j+1}^2; Z_{j+1})$ (the same for $(n; \ell_n^1, \ell_n^2; Z_n)$)
the algorithm checks whether or not the box $(j+1)$-$B_{i_1, i_2, i_3}$,
where $\ell_{j+1}^1 \in I_{i_1}^1$, $\ell_{j+1}^2 \in I_{i_2}^2$, and $Z_{j+1} \in I_{i_3}^3$,
is empty or not;
if it is empty, then the quadruple is saved in the box,
otherwise the box is updated to save the one having a smaller $\ell_{j+1}^1$ between the old and the new quadruples.

\ 

\noindent{\bf Step 3.} \ 
For each box $n$-$B_{i_1, i_2, i_3}$, where $(i_1, i_2, i_3) \in (\{0\} \cup [r_1]) \times [r_2] \times [r_3]$,
if there is a saved quadruple $(n; \ell_n^1, \ell_n^2; Z_n)$, then there is a constrained full schedule with the total weighted completion time $Z_n$.
Here the constraint is that the job $J_a$ starts processing at time $T_2$.
The algorithm {\sc Approx}$(a, \epsilon)$ scans through all the boxes associated with $n$,
and returns the quadruple having the smallest total weighted completion time, denoted as $Z_n^a$.
The corresponding constrained full schedule can be backtracked.

\begin{theorem}
\label{thm08}
Algorithm {\sc Approx}$(a, \epsilon)$ is a $(1 + \epsilon)$-approximation for the rescheduling problem $(1, h_1 \mid \Delta_{\max}\le k \mid \sum_{j=1}^n w_j C_j)$
under the constraint that the $J_a$ starts processing at time $T_2$;
its time complexity is
$O\left(\frac{n^4}{\epsilon^3} \log T_1 \log(T_2 + P) \log(w_{\max}P + W T_2)\right)$.
\end{theorem}
\begin{proof}
The proof of the performance ratio is done by induction.
Assume that $(a, \ell_a^{1*}, \ell_a^{2*}, Z_a^*) \rightarrow (a+1, \ell_{a+1}^{1*}, \ell_{a+1}^{2*}, Z_{a+1}^* )
\rightarrow \ldots \rightarrow (j_3, \ell_{j_3}^{1*}, \ell_{j_3}^{2*}, Z_{j_3}^*) \rightarrow (n, \ell_n^{1*}, \ell_n^{2*}, Z_n^*)$
is the path of quadruples computed by the exact algorithm in Theorem~\ref{thm07}
that leads to the constrained optimal solution for the problem $(1, h_1 \mid \Delta_{\max}\le k \mid \sum_{j=1}^n w_j C_j)$.
The induction statement is for each $j$, $a \le j \le j_3$ or $j = n$,
there is a quadruple $(j; \ell_j^1, \ell_j^2; Z_j)$ saved by the algorithm {\sc Approx$(a, \epsilon)$},
such that $\ell_j^1 \le \ell_j^{1*}$, $\ell_j^2 \le \ell_j^{2*} \delta^j$, $Z_j \le Z_j^* \delta^j$.

The base case is $j = a$, and the statement holds since there is only one partial schedule on the first $a$ jobs,
described as $(a, \ell_a^1, \ell_a^1, Z_a)$ in Eq.~(\ref{eq12}).
We assume the induction statement holds for $j$, where $a \le j \le j_3$, that is,
there is a quadruple $(j; \ell_j^1, \ell_j^2; Z_j)$ saved by the algorithm {\sc Approx$(a, \epsilon)$},
such that
\begin{equation}
\label{eq18}
\ell_j^1 \le \ell_j^{1*}, \ell_j^2 \le \ell_j^{2*} \delta^j, \mbox{ and } Z_j \le Z_j^* \delta^j.
\end{equation}

When $j < j_3$, from this particular quadruple $(j; \ell_j^1, \ell_j^2; Z_j)$,
we continue to assign the job $J_{j+1} \in {\cal J}_2 \cup {\cal J}_3$ as in the exact algorithm in Subsection 4.1.
In Case 1 where $J_{j+1}$ is added in the later schedule to obtain a partial schedule described as in Eq.~(\ref{eq13}),
we have 
\begin{equation}
\label{eq19}
\ell_{j+1}^1 = \ell_j^1, \ell_{j+1}^2 = \ell_j^2 + p_{j+1}, \mbox{ and } Z_{j+1} = Z_j + w_{j+1} (\ell_j^2 + p_{j+1})).
\end{equation}
Assume this quadruple falls in the box $j$-$B_{i_1, i_2, i_3}$,
then the quadruple $(j+1; \hat{\ell}_{j+1}^1, \hat{\ell}_{j+1}^2; \hat{Z}_{j+1})$ saved in this box by the algorithm {\sc Approx$(a, \epsilon)$} must have
\begin{equation}
\label{eq20}
\hat{\ell}_{j+1}^1 \le \ell_{j+1}^1, \hat{\ell}_{j+1}^2 \le \ell_{j+1}^2 \delta, \mbox{ and } \hat{Z}_{j+1} \le Z_{j+1} \delta.
\end{equation}
If $(j; \ell_j^{1*}, \ell_j^{2*}; Z_j^*)$ leads to $(j+1; \ell_{j+1}^{1*}, \ell_{j+1}^{2*}; Z_{j+1}^*)$ also by
adding $J_{j+1}$ in the later schedule, then we have
\begin{equation}
\label{eq21}
\ell_{j+1}^{1*} = \ell_j^{1*}, \ell_{j+1}^{2*} = \ell_j^{2*} + p_{j+1}, \mbox{ and } Z_{j+1}^* = Z_j^* + w_{j+1} (\ell_j^{2*} + p_{j+1})).
\end{equation}
From Eqs.~(\ref{eq18}--\ref{eq21}) and $\delta > 1$, we have
\begin{equation}
\label{eq22}
\hat{\ell}_{j+1}^1 \le \ell_{j+1}^{1*}, \hat{\ell}_{j+1}^2 \le \ell_{j+1}^{2*} \delta^{j+1}, \mbox{ and } \hat{Z}_{j+1} \le Z_{j+1}^* \delta^{j+1}.
\end{equation}

Similarly, in Case 2 where $J_{j+1}$ is added in the earlier schedule without inserting a machine idling period
to obtain a feasible partial schedule described as in Eq.~(\ref{eq14}),
we can show that if $(j; \ell_j^{1*}, \ell_j^{2*}; Z_j^*)$ leads to $(j+1; \ell_{j+1}^{1*}, \ell_{j+1}^{2*}; Z_{j+1}^*)$ also in this way,
then there is a saved quadruple $(j+1; \hat{\ell}_{j+1}^1, \hat{\ell}_{j+1}^2; \hat{Z}_{j+1})$ by the algorithm {\sc Approx$(a, \epsilon)$}
such that Eq.~(\ref{eq22}) holds.

In Case 3 where $\ell_j^1 + p_{j+1} < C_{j+1}(\pi^*) - k \le T_1$, $J_{j+1}$ is added in the earlier schedule,
and a non-empty machine idling period of length $C_{j+1}(\pi^*) - (\ell_j^1 + p_{j+1}) - k$ is inserted right before it,
the algorithm continuously processes succeeding jobs up to $J_{j_3}$ in the earlier schedule,
and lastly process all the remaining jobs in the later schedule.
When $j = j_3$, the algorithm completes the quadruple by assigning all the jobs of ${\cal J}_4$ in the later schedule, starting at time $\ell_j^2$,
to obtain a full schedule described as in Eq.~(\ref{eq15}).
We can similarly show that if $(j; \ell_j^{1*}, \ell_j^{2*}; Z_j^*)$ leads to $(n; \ell_n^{1*}, \ell_n^{2*}; Z_n^*)$ directly,
or if $(j; \ell_j^{1*}, \ell_j^{2*}; Z_j^*)$ leads to $(j+1; \ell_{j+1}^{1*}, \ell_{j+1}^{2*}; Z_{j+1}^*)$, and so on,
then eventually to $(n; \ell_n^{1*}, \ell_n^{2*}; Z_n^*)$,
then there is also a saved quadruple $(n; \hat{\ell}_n^1, \hat{\ell}_n^2; \hat{Z}_n)$ by the algorithm {\sc Approx$(a, \epsilon)$}
such that Eq.~(\ref{eq22}) holds with $n$ replacing $j+1$.
We therefore finish the proof of the induction statement.

Lastly we use the inequality $(1 + \epsilon/2n)^n \le 1 + \epsilon$ for $0 < \epsilon < 1$ to bound the ratio $\delta^n$.
That is, the total weighted completion time of the output full schedule by the algorithm {\sc Approx$(a, \epsilon)$} is
\[
Z_n^a \le Z_n^* \delta^n \le (1 + \epsilon) Z_n^*.
\]

For the time complexity,
note that for each $j$ such that $a \le j \le j_3$ or $j = n$, there are $O(r_1 r_2 r_3)$ boxes associated with it;
for a saved quadruple $(j; \ell_j^1, \ell_j^2; Z_j)$, it takes $O(1)$ time to assign the job $J_{j+1}$,
leading to at most three new quadruples $(j+1; \ell_{j+1}^1, \ell_{j+1}^2; Z_{j+1})$, each is used to update the corresponding box associated with $j+1$.
It follows that the total running time is $O(n r_1 r_2 r_3)$.
Note from the definitions of $r_i$'s, and $\log \delta = \log(1 + \epsilon/2n) \ge \epsilon/4n$,
that $O(n r_1 r_2 r_3) \subseteq O(n^4/\epsilon^3 \log T_1 \log(T_2 + P) \log(w_{\max}P + W T_2))$.
\end{proof}

Enumerating all possible $O(n)$ choices of $J_a$, that is calling the algorithm {\sc Approx}$(a, \epsilon)$ $O(n)$ times,
the full schedule with the minimum $Z_n^a$ is returned as the final solution.
The overall algorithm is denoted as {\sc Approx}$(\epsilon)$.

\begin{corollary}
\label{coro09}
Algorithm {\sc Approx}$(\epsilon)$ is a $(1 + \epsilon)$-approximation for the problem $(1, h_1 \mid \Delta_{\max}\le k \mid \sum_{j=1}^n w_j C_j)$;
its time complexity is
$O\left(\frac{n^5}{\epsilon^3} \log T_1 \log(T_2 + P) \log(w_{\max}P + W T_2)\right)$.
\end{corollary}

\subsection{An FPTAS for the general case}
In this subsection, we derive an FPTAS for the general rescheduling problem $(1, h_1 \mid \Delta_{\max}\le k \mid \mu\Delta_{\max} + \sum_{j=1}^n w_j C_j)$
by invoking multiple times the approximation algorithm {\sc Approx$(a, \epsilon)$} for the special case $\mu = 0$.

The job $J_a$ starts processing at time $T_2$.
It follows that the maximum time deviation is bounded $\Delta_a \le \Delta_{\max} \le k$,
and we use $(1+\epsilon)$ to split the interval $[\Delta_a, k]$ into subintervals
\[
I_1^4 = [\Delta_a, \Delta_a (1+\epsilon)], \ I_2^4 = (\Delta_a (1+\epsilon), \Delta_a (1+\epsilon)^2], \ \ldots, \ I_{r_4}^4 = (\Delta_a (1+\epsilon)^{r_4}, k],
\]
where $r_4 \triangleq \lfloor \log_{1+\epsilon} (k /\Delta_a)\rfloor$.

Let $k_i = \Delta_a (1+\epsilon)^i$, for $i = 1, 2, \ldots, r_4$, and $k_{r_4+1} = k$.
Using $k_i$ as the upper bound on the maximum time deviation, that is, $k_i$ replaces $k$,
our algorithm calls {\sc Approx$(a, \epsilon)$} to solve the problem $(1, h_1 \mid \Delta_{\max}\le k_i \mid \sum_{j=1}^n w_j C_j)$,
by returning a constrained full schedule of the total weighted completion time within $(1 + \epsilon)$ of the minimum.
We denote this full schedule as $\sigma(a, k_i)$.
The final output full schedule is the one of the minimum objective function value among $\{\sigma(a, k_i) \mid 1 \le a \le j_1, \ 1 \le i \le r_4+1\}$,
which is denoted as $\sigma^\epsilon$ with the objective function value $Z^\epsilon$.
We denote our algorithm as {\sc Approx$^\mu(\epsilon)$}.

\begin{theorem}
\label{thm10}
Algorithm {\sc Approx}$^\mu(\epsilon)$ is a $(1 + \epsilon)$-approximation for the general rescheduling problem
$(1, h_1 \mid \Delta_{\max}\le k \mid \mu\Delta_{\max} + \sum_{j=1}^n w_j C_j)$;
its time complexity is\\
$O\left(\frac{n^5}{\epsilon^4} \log T_1 \log(T_2 + P) \log(w_{\max}P + W T_2) \log k\right)$.
\end{theorem}
\begin{proof}
Let $\sigma^*$ denote an optimal schedule for the problem with the objective function value $Z^*$,
satisfying Lemmas~\ref{lemma01} and \ref{lemma02}.
Suppose the job starts processing at time $T_2$ in $\sigma^*$ is $J_a$, and the maximum time deviation in $\sigma^*$ is $\Delta_{\max}^* = k^* \le k$.
We thus conclude that $\sigma^*$ is also a constrained optimal schedule for the problem by replacing the upper bound $k$ with $k^*$, i.e.,
$(1, h_1 \mid \Delta_{\max}\le k^* \mid \mu\Delta_{\max} + \sum_{j=1}^n w_j C_j)$,
under the constraint that $J_a$ starts processing at time $T_2$.

Consider the value $k_i = \Delta_a (1+\epsilon)^i$ such that $\Delta_a (1+\epsilon)^{i-1} \le k^* \le \Delta_a (1+\epsilon)^i$.

Clearly, for the schedule $\sigma(a, k_i)$ found by the algorithm {\sc Approx$(a, \epsilon)$}
to the problem $(1, h_1 \mid \Delta_{\max}\le k_i \mid \sum_{j=1}^n w_j C_j)$,
it has the total weighted completion time $Z(\sigma) \le (1 + \epsilon) Z(\sigma^*)$
and the maximum time deviation $\Delta_{\max}(\pi^*, \sigma) \le k_i \le k^* (1+\epsilon)$. 
It follows that
\[
Z^\epsilon \le \mu\Delta_{\max}(\pi^*, \sigma) + Z(\sigma) \le \mu k^* (1+\epsilon) + (1 + \epsilon) Z(\sigma^*) = (1 + \epsilon) Z^*.
\]

Note that we call the algorithm {\sc Approx$(a, \epsilon)$} for all possible values of $a$,
and for each $a$ we call the algorithm $r_4 + 1 = O(\frac 1{\epsilon} \log k)$ times (using the inequality $\log(1 + \epsilon) \ge \frac 12 \epsilon$).
Thus from Theorem~\ref{thm08} the total running time of the algorithm {\sc Approx$^\mu(\epsilon)$} is in
\[
O\left(\frac{n^5}{\epsilon^4} \log T_1 \log(T_2 + P) \log(w_{\max}P + W T_2) \log k\right).
\]
\end{proof}

\section{Concluding remarks}
We investigated a rescheduling problem where a set of jobs has already been scheduled to minimize the total weighted completion time on a single machine,
but a disruption causes the machine become unavailable for a given time interval.
The production planner needs to reschedule the jobs without excessively altering the originally planned schedule.
The degree of alteration is measured as the maximum time deviation for all the jobs between the original and the new schedules.
We studied a general model where the maximum time deviation is taken both as a constraint and as part of the objective function.
We presented a pseudo-polynomial time exact algorithm based on dynamic programming and an FPTAS.
We remark that the FPTAS is developed from another slower exact algorithm for the special case
where the maximum time deviation is not taken as part of the objective function,
which gives us room to properly sample the maximum time deviation.

\section*{Acknowledgments}
W.L. was supported by K. C. Wong Magna Found in Ningbo University, the China Scholarship Council (Grant No. 201408330402),
and the Ningbo Natural Science Foundation (2016A610078).
T.L. was supported by Brain Canada and NSF China (Grant Nos. 61221063 and 71371129).
R.G. was supported by NSERC Canada.
G.L. was supported by NSERC Canada and NSF China (Grant No. 61672323).

\newpage

\newpage
\appendix
\section{Proof of Lemma~\ref{lemma01}}
\begin{proof}
(Item (b) of Lemma~\ref{lemma01})
By contradiction, assume $(J_i, J_j)$ is the first pair of jobs for which $J_i$ precedes $J_j$ in $\pi^*$, i.e. $p_i/w_i \le p_j/w_j$,
but $J_j$ immediately precedes $J_i$ in the later schedule of $\sigma^*$.
Let $\sigma'$ denote the new schedule obtained from $\sigma^*$ by swapping $J_j$ and $J_i$.
If $C_j(\sigma') \ge C_j(\pi^*)$, then $C_i(\sigma') \ge C_i(\pi^*)$ too,
and thus $\Delta_j(\sigma', \pi^*) \le \Delta_i(\sigma', \pi^*) < \Delta_i(\sigma^*, \pi^*)$ due to $p_j > 0$;
if $C_j(\sigma') < C_j(\pi^*)$ and $C_i(\sigma') \le C_i(\pi^*)$,
then $\Delta_i(\sigma', \pi^*) \le \Delta_j(\sigma', \pi^*) < \Delta_j(\sigma^*, \pi^*)$ due to $p_i > 0$;
lastly if $C_j(\sigma') < C_j(\pi^*)$ and $C_i(\sigma') > C_i(\pi^*)$,
then $\Delta_i(\sigma', \pi^*) < \Delta_i(\sigma^*, \pi^*)$ and $\Delta_j(\sigma', \pi^*) < \Delta_j(\sigma^*, \pi^*)$.
That is, $\sigma'$ is also a feasible reschedule.

Furthermore, the weighted completion times contributed by $J_i$ and $J_j$ in $\sigma'$ is no more than those in $\sigma^*$,
implying the optimality of $\sigma'$.
If follows that, if necessary,
after a sequence of job swappings we will obtain an optimal reschedule in which the jobs in the earlier schedule are in the same order as they appear in $\pi^*$.
\end{proof}

\section{Proof of Lemma~\ref{lemma02}}
\begin{proof}
Item (a) is a direct consequence of Lemma~\ref{lemma01}, since the jobs in the earlier schedule are in the same order as they appear in $\pi^*$,
which is the WSPT order.
That is, for the job $J_j$ in the earlier schedule of $\sigma^*$, if all the jobs $J_i$, for $i = 1, 2, \ldots, j-1$, are also in the earlier schedule,
then $C_j(\sigma^*) = C_j(\pi^*)$;
otherwise, $C_j(\sigma^*) < C_j(\pi^*)$.

For item (b), assume the machine idles before processing the jobs $J_{j_1}$ and $J_{j_2}$, with $J_{j_1}$ preceding $J_{j_2}$ in the earlier schedule.
We conclude from Lemma~\ref{lemma01} and item (a) that if $\Delta_{j_2} < \Delta_{j_1} \le \Delta_{\max}$
(or $\Delta_{j_1} < \Delta_{j_2} \le \Delta_{\max}$, respectively),
then moving the starting time of the job $J_{j_2}$ ($J_{j_1}$, respectively) one unit ahead
will maintain the maximum time deviation and decrease the total weighted completion time, which contradicts the optimality of $\sigma^*$.
It follows that we must have $\Delta_{j_1} = \Delta_{j_2}$;
in this case, if there is any job $J_j$ with $j_1 < j < j_2$ in the later schedule,
it can be moved to the earlier schedule to decrease the total weighted completion time, which again contradicts the optimality of $\sigma^*$.
Therefore, there is no job $J_j$ with $j_1 < j < j_2$ in the later schedule,
which together with $\Delta_{j_1} = \Delta_{j_2}$ imply that the machine does not idle before processing the job $J_{j_2}$.

Item (c) is clearly seen for the same reason used in the last paragraph,
that firstly their time deviations have to be the same,
and secondly if this time deviation is less than $\Delta_{\max}$,
one can then move their starting time one unit ahead  to decrease the total weighted completion time while maintaining the maximum time deviation,
thus contradicting the optimality of $\sigma^*$.

Item (d) is implies by Item (c), given that all the job processing times are positive.

Let the job in the earlier schedule right after the idle time period be $J_j$.
clearly, there is a job $J_a$ with $S_a(\pi^*) < S_j(\sigma^*)$ in the later schedule of $\sigma^*$, due to the machine idling.
The time deviation for $J_a$ is $\Delta_a > T_2 - S_a(\pi^*)$.
If $S_j(\pi^*) < T_2$, then $\Delta_{\max} = \Delta_j < T_2 - S_j(\sigma^*) < T_2 - S_a(\pi^*) < \Delta_a$, a contradiction.
This proves item (e) that $S_j(\pi^*) \ge T_2$.

Item (f) is clearly seen from Lemma~\ref{lemma01} and the optimality of $\sigma^*$,
that if the machine idles then it can start processing the jobs after the idling period earlier to decrease the total weighted completion time,
while maintaining or even decreasing the maximum time deviation.

From item (f) and the second part of Lemma~\ref{lemma01}, we see that the deviation times of the jobs in the later schedule are non-increasing.
Therefore, item (g) is proved.
\end{proof}

\end{document}